\documentclass[runningheads]{llncs}

\usepackage[T1]{fontenc}
\usepackage{graphicx}
\usepackage{amsmath}
\usepackage{amssymb}
\usepackage{xcolor}
\usepackage{textcomp}
\usepackage{listings}
\usepackage{courier}
\usepackage{extarrows}
\usepackage{algorithm2e}
\usepackage{algpseudocode}
\usepackage[utf8]{inputenc}
\usepackage{threeparttable}
\usepackage{mathpartir}
\usepackage{appendix}

\definecolor{rustkeyword}{RGB}{0,0,180}
\definecolor{rusttype}{RGB}{140,0,0}
\definecolor{rustcomment}{RGB}{0,110,0}
\definecolor{ruststring}{RGB}{160,0,160}

\newcommand{\Crate}{\mathcal{C}}        

\newcommand{\Type}{\text{Type}}

\newcommand{\dom}{\mathsf{dom}}

\newcommand{\Canon}{\mathsf{Canon}}

\newcommand{\TokObs}{\mathsf{Obs}}
\newcommand{\FreshOK}{\mathsf{FreshOK}}

\newcommand{\CapOK}{\mathsf{CapOK}}

\newcommand{\unify}{\mathrel{\dot{=}}}

\newcommand{\Cfg}{\mathsf{Cfg}}
\newcommand{\Fire}{\mathsf{Fire}}
\newcommand{\SolveGuard}{\mathsf{SolveGuard}}

\newcommand{\fvL}{\mathsf{fv}_L}

\lstdefinelanguage{Rust}{
  keywords={fn, let, mut, struct, enum, impl, return, match, if, else, while, for, in, move, pub, use, mod, crate, ref, where, unsafe},
  keywordstyle=\color{rustkeyword}\bfseries,
  ndkeywords={String, Vec, Option, Result, u32, i32, bool, Self, Box, Rc, Arc},
  ndkeywordstyle=\color{rusttype}\bfseries,
  comment=[l]{//},
  commentstyle=\color{rustcomment}\itshape,
  stringstyle=\color{ruststring},
  sensitive=true,
  basicstyle=\ttfamily\small,
  breaklines=true,
  frame=single,
  numbers=left,
  numberstyle=\tiny\color{gray},
  captionpos=b,
  escapeinside={(*@}{@*)}, 
}

\begin{document}
\title{A Synthesis Method of Safe Rust Code Based on Pushdown Colored Petri Nets}
\titlerunning{Safe Rust Synthesis Based on Pushdown Colored Petri Nets}
%
%
\author{Kaiwen Zhang\inst{1}\orcidID{0000-0003-4573-8968} \and
Guanjun Liu\inst{1}\orcidID{0000-0002-7523-4827}}
\authorrunning{Kaiwen Zhang et al.}
%
\institute{Tongji University, Shanghai, China \email{\{zhangkw,liuguanjun\}@tongji.edu.cn}}
\maketitle              
\begin{abstract}
Safe Rust guarantees memory safety through strict compile-time constraints: \textit{ownership} can be transferred, \textit{borrowing} can temporarily guarantee either shared read-only or exclusive write access, and ownership and borrowing are scoped by \textit{lifetime}. 
Automatically synthesizing correct and safe Rust code is challenging, as the generated code must not only satisfy ownership, borrowing, and lifetime constraints, but also meet \textit{type} and \textit{interface} requirements at compile time.
This work proposes a synthesis method based on our newly defined Pushdown Colored Petri Net (PCPN) that models these compilation constraints directly from public API signatures to synthesize valid call sequences. Token colors encode dynamic resource states together with a scope level indicating the lifetime region in which a \textit{borrow} is valid. The pushdown stack tracks the entering or leaving of \textit{lifetime parameter} via pushing and popping tokens. A transition is enabled only when type matching and interface obligations both hold and the required resource states are available.
Based on the bisimulation theory, we prove that the enabling and firing rules of PCPN are consistent with the \textit{compile-time check} of these three constraints. We develop an automatic synthesis tool based on PCPN and the experimental results show that the synthesized codes are all correct.

\end{abstract}

\section{Introduction}
Safe Rust turns many programming decisions into compile time checks \cite{matsakis2014rust}. A client program compiles only when it respects \textit{ownership} transfer (i.e., moves), \textit{borrowing} exclusivity, and \textit{lifetime} scoping \cite{ho2022aeneas}. These checks are local to a compilation unit and enforced without running our concerned program \cite{jung2019stacked,jung2017rustbelt,villani2025tree}. However, they are central obstacles for automated code generation. A candidate call sequence of automatically-generated code can be \textit{type} correct at a coarse level but be rejected at compile time because a value is used after its ownership has been moved, or a mutable borrow overlaps with another one, or a generic constraint has no valid instantiation. 

Petri nets are a kind of formal models that are used to analyze  reachability problems in concurrent systems \cite{2024ITSMC..54.3908D,kosaraju1982decidability,rackoff1978covering,reps1995precise,mengchu1}, while the compile-time check problem of automatically-generated code is exactly represented the reachability problem of Petri nets. Tokens represent type resources in a program, transitions represent actions such as function call, and reachability captures whether a sequence of actions can be executed or not. Every public Rust API is  viewed as an atomic transformer specified solely by its signature and then modeled by a transition of a Petri net. Consequently, if a transition sequence is firable and satisfies that a multiset of available values is transferred into a multiset that contains the desired result \cite{mengchu2}, then we can backtrack this firable sequence into a compilable Rust code.

Classical place/transition nets are not expressive enough for \textit{safe Rust}. The first limitation is Rust typing. Rust signatures carry parametric types, references, and projections that must match across call sites \cite{deline2004typestates,feng2017component}. The second limitation is permission evolution. Borrowing does not merely read a token, it changes what can be done with the referent while the borrow is live. The third limitation is lifetime scope. A borrow is valid only within a lifetime region, and regions must respect \textit{outlives relationship} (i.e., ordering constraints between lifetimes) induced by the program structure. In essence, classical nets treat tokens as identity-less  units, making them incapable of tracking the generic types, fluid permissions of references, or hierarchical outlives relationships that define Rust's safety guarantees \cite{astrauskas2019leveraging,jung2019stacked,jung2017rustbelt,vanhattum2022verifying,villani2025tree}.

We address these issues by constructing a Pushdown Colored Petri Net (PCPN) from a compiler extracted environment $\Sigma(\Crate)$ \cite{5533237,jensen1987coloured,jensen2007coloured}. 
Places are indexed by concrete types and by a small permission capability that approximates the borrow checker view of a live \textit{binding}. Transitions come from two sources. The first source is the set of callable items extracted from signatures, including functions and methods. The second source is a fixed family of structural steps that model ownership, bounded duplication, borrow creation and discharge, projections, and reborrowing. 
Token colors carry value identifiers together with instantiations of type and \textit{lifetime parameters}. Internally created references carry an explicit region label that denotes the concrete region in which the borrow remains valid. A pushdown stack stores a well scoped record of \textit{outstanding borrows} \cite{alur2004visibly,10.1145/2568225.2568311}, such that borrow creation and termination correspond to push and pop actions, respectively. A transition can fire only when its input tokens match its expected \textit{type schemes} and when all interface obligations are satisfied. These obligations include trait bounds, associated type equalities, and outlives relationship over lifetime parameters, which are checked during guarded enabling by using finite ground facts extracted from $\Sigma(\Crate)$.

To justify this encoding, we define a small-step signature-induced resource semantics. The semantics abstracts each API call as a state transformer admitted by the signature environment. It also gives explicit rules for ownership transfer, borrowing, reborrowing, and lifetime discharge. A strong bisimulation is established between this semantics and the firing semantics of a PCPN. As a consequence, an API call trace is compilable exactly when it corresponds to a firing sequence that reaches a closed configuration whose borrow stack is empty. To ensure decidability, we obtain an effective finite search space by fixing explicit bounds on token counts and borrow nesting depth \cite{bouajjani1997reachability,rackoff1978covering,solar2006combinatorial}. By quotienting configurations via a canonicalization function, the reachability graph is saturated through worklist exploration, allowing firing sequences to be backtracked and emitted as syntactically valid Rust snippets. We implement a prototype tool based on PCPN, which is publicly available at \url{https://github.com/kevindadi/RustSynth}.

The new contributions we intend to make are as follows:
\begin{enumerate}
  \item We formulate a signature-induced semantics for safe Rust that captures ownership transfer, borrowing, and lifetime ordering. Notably, this semantics enables generic constraint solving without requiring the inspection of function bodies;
  \item We provide a method to directly construct PCPN is constructed directly from $\Sigma(\Crate)$ by utilizing places indexed by types and capabilities, alongside a dedicated borrow stack to enforce properly nested lifetime regions;
  \item We integrate type matching and constraint discharge into guarded enabling through a unified judgment. This mechanism can perform unification while verifying trait bounds, associated types, and outlives relationships;
  \item We establish step correspondence and strong bisimulation between the resource semantics and PCPN, yielding an exact reachability characterization of compilable API call traces; and 
  \item Under explicit bounds, we develop a finite canonical reachability graph. This structure supports effective witness  (i.e., reconstructing valid API call sequences) extraction and the emission of trusted, syntactically valid Rust snippets.
\end{enumerate}

Section~\ref{sec:background} reviews the Rust features and the signature induced semantic interface used by our net construction. 
Section~\ref{sec:method} defines PCPN and its guarded enabling.
Section~\ref{sec:method:construct} gives the construction process of PCPN and the bisimulation theorem.
Section~\ref{sec:analysis} develops bounded reachability, canonical quotienting, and witness oriented search.
Section~\ref{sec:conlusion} concludes this work.

\section{Background}
\label{sec:background}

\subsection{Rust ownership, borrowing, and lifetimes}
\label{sec_background_rust}
This paper targets Safe Rust, a programming paradigm that enforces memory safety through a combination of static analysis and \textit{affine} type systems. Unlike conventional languages, Rust manages resources without a garbage collector by relying on the principles of ownership, borrowing, and lifetimes.

At the core of Safe Rust is the concept of ownership, which dictates that every value in Rust has a unique owner, known as a binding. For types that do not implement the \texttt{Copy} trait, such as strings or complex data structures, Rust enforces move semantics. When such a value is passed as an argument or assigned to a new variable, its ownership is transferred to the destination, and the source binding is immediately invalidated. This prevents multiple paths from attempting to manage the same memory simultaneously, effectively modeling values as unique resources. In contrast, \texttt{Copy} types allow for implicit duplication, where the value is cloned rather than moved, leaving the source binding active and usable \cite{girard1987linear,tov2011practical,wadler1990linear}.

While ownership ensures clear resource disposal, borrowing provides a mechanism for temporary access to a value through references without transferring its ownership. To maintain safety, Rust imposes a strict aliasing discipline akin to a readers-writer lock: a value may have either any number of shared references (\texttt{\&T}), which grant read-only access, or exactly one mutable reference (\texttt{\&mut T}), which allows exclusive write access. The compiler’s borrow checker ensures that these two states are mutually exclusive. For instance, no value can be modified while it is being read by others. This rule prevents data races at compile time and ensures that references always point to valid memory.

The validity of these references is further constrained by \textit{lifetimes}. A lifetime defines the region of code in which a borrow remains valid. A fundamental invariant is that a reference must never outlive the referent it points to. Although modern Rust utilizes \textit{None Lexical Lifetimes} (NLL) to allow borrows to end precisely at their last point of use, our synthesis approach adopts a more structured discipline. By inserting explicit \texttt{drop} calls to terminate borrows, we ensure that their lifetimes are properly nested. This structured approach is sound under the broader NLL rules and provides a predictable execution trace for code generation.

Building upon this borrowing mechanism, reborrowing allows a program to derive a new, shorter-lived reference from an existing one. This is particularly common in method calls where a mutable reference is temporarily downgraded or restricted to a specific field. Crucially, the parent reference is suspended and becomes inaccessible while the child reborrow is live. Once the child reference is dropped, the parent reference regains its original permissions. This hierarchical suspension and resumption of access is essential for representing the complex call sequences found in standard Rust libraries as valid compilable code.

Fig.~\ref{fig:rust_core_rules} illustrates ownership transfer, where $s$ is invalidated after moving to $t$, and the borrowing exclusivity rule, which allows the mutable borrow $m$ only after the shared borrow $r$ is explicitly dropped. It also shows explicit termination of borrows.
\begin{figure}[htbp]
\centering
\begin{lstlisting}[language=Rust]
let s = R(1); 
let t = s;        // Ownership moved from s to t
// let u = s;     // Illegal: s is now moved
let mut x = 0;
let r = &x;        // Shared borrow starts
// let m = &mut x; // Illegal: mutable borrow conflicts with r
drop(r);           // Explicit end of borrow
let m = &mut x;    // Legal: exclusive access granted
\end{lstlisting}
\caption{A move and the shared and mutable borrow rules.}
\label{fig:rust_core_rules}
\end{figure}

\subsection{A Small Signature-Driven Core Language}
\label{sec_background_core}
To bridge the gap between Rust source code and Petri net transitions  \cite{jensen1987coloured,liu2022petri}, we define a core syntax that distills API interactions into a sequence of atomic resource transformers. We represent the synthesized client code as a program $\textsf{prog}$, consisting of a sequence of statements $\textsf{stmt}$ defined as follows:

\[
\begin{aligned}
s \in \textsf{stmt} \triangleq\ \ &
\textsf{let } x \textsf{ = Call}(f, \bar{x}) \mid \textsf{drop}(x) \\
\mid\ & \textsf{let } r \textsf{ = Borrow}^{\{shr, mut\}}(x) \mid \textsf{let } r' \textsf{ = Reborrow}^{\{shr, mut\}}(r) \\
\mid\ & \textsf{let } y \textsf{ = ProjMove}(x, f) \mid \textsf{let } r' \textsf{ = ProjRef}(r, f) \mid \textsf{let } y \textsf{ = Deref}(r)
\end{aligned}
\]
Each construct abstracts a fundamental transition in the Rust resource model. \textsf{Call} represents function or method invocations where ownership effects are determined by $\Sigma(\mathcal{C})$. \textsf{Borrow} and \textsf{Reborrow} model the creation of new reference tokens. The latter specifically accounts for the temporary suspension of parent references. \textsf{ProjMove} and \textsf{ProjRef} handle access to struct fields or tuple components, ensuring that the synthesizer can reason about sub-resources within complex data structures. \textsf{Deref} accounts for duplicating values of \texttt{Copy} types through references, yielding an independent owned value without affecting the reference's usability. In contrast, \textsf{Reborrow} models the derivation of a new reference from an existing one; this process imposes a hierarchical dependency where the parent reference is temporarily suspended until the child reference is dropped, thereby maintaining strict aliasing invariants during nested calls.
Finally, \textsf{drop}(x) serves a dual purpose: for owned values, it models resource deallocation; for references, it acts as an explicit signal to terminate a borrow's lifetime. This explicit termination is critical for our formal model to unfreeze a owner and restore its access permissions, ensuring that the synthesized code respects the LIFO (Last-In-First-Out) stack discipline of the underlying PCPN \cite{matsakis2014rust}. 

This syntax is sufficient for API-level reasoning \cite{gabbay2002new,weiss2019oxide}, as all library interactions are encapsulated in \textsf{Call} statements. Ownership effects are induced by the parameter types (e.g., \texttt{Copy} vs. non-\texttt{Copy}), while borrow effects are managed through explicit reference creation and termination. The set of live regions at any point in the synthesized program is therefore deterministically governed by the set of active values and their corresponding borrow structure.

\section{Pushdown Colored Petri Nets}
\label{sec:method}
This section formalizes our synthesis approach. We bridge the gap between Rust's rich static semantics and the reachability analysis of Petri nets by transforming the compiler-extracted signature environment $\Sigma(\mathcal{C})$ into a PCPN. The core challenge lies in mapping Rust's potentially infinite type space due to generics and lifetimes into a finite net structure. We address this by distinguishing between \textit{type schemes} appearing in signatures and \textit{ground types} indexing Petri net places. The resulting PCPN models the multiset of live values as tokens and enforces borrowing discipline via a pushdown stack. A valid firing sequence in this net corresponds to a well-typed, borrow-safe Rust code snippet.

\subsection{Problem Setting and Type Formalization}
\label{sec:method:setting}
Our synthesis target is defined by a finite set $\mathcal{F}$ of callable items extracted from $\Sigma(\mathcal{C})$. Each item $f \in \mathcal{F}$ is associated with a signature containing generic type parameters, lifetime constraints, and trait obligations. Formally, let $\mathsf{GTy}$ be the set of monomorphic types (e.g., \texttt{u8} and \texttt{bool}) and nominal struct types extracted from the crate. Let $\mathsf{TVar}$ and $\mathsf{LVar}$ be countable sets of generic type variables and lifetime variables, respectively. To handle lifetime instantiation dynamically during synthesis, we introduce a special reserved lifetime variable $\ell_{\bullet} \notin \mathsf{LVar}$, which serves as a runtime hook to carry concrete region labels generated by the net inside token colors.

\begin{definition}[Type Schemes]
\label{def:pcpn:types}
The set of type schemes $\mathsf{Ty}$ represents the polymorphic types found in API signatures. It is generated by the grammar:
\begin{align*}
\tau \in \mathsf{Ty} ::=~&
g
\mid \alpha
\mid \tau_0\langle \tau_1,\dots,\tau_n\rangle
\mid (\tau_1,\dots,\tau_n)
\mid [\tau]
\mid \mathsf{Ref}^q_{\ell}(\tau) \\
&\mid \mathsf{Assoc}(\tau,\mathsf{Tr},A)
\mid \mathsf{Field}(\tau,f)
\end{align*}
where $g\in\mathsf{GTy}$, $\alpha\in\mathsf{TVar}$, $\ell \in \mathsf{LVar} \cup \{\ell_{\bullet}\}$, and $q \in \{\mathsf{shr}, \mathsf{mut}\}$. 
Constructs $\mathsf{Assoc}$ and $\mathsf{Field}$ are symbolic representations of associated types and field projections, resolved during unification.
\end{definition}

While $\mathsf{Ty}$ captures the expressiveness of signatures, it is too large to index the places of a Petri net. To ensure that PCPN remains finitely branching, we restrict the state space to a finite universe of \textit{ground types}. Let $\widehat{\mathsf{Ty}} \subset \mathsf{Ty}$ be a finite set of instantiated, lifetime types. It is defined as the least set satisfying: (i) Base types $\mathsf{GTy} \subseteq \widehat{\mathsf{Ty}}$; (ii) Compound types (tuples and slices) are in $\widehat{\mathsf{Ty}}$ only if their components are in $\widehat{\mathsf{Ty}}$; (iii) Reference types $\mathsf{Ref}^q_{\ell}(\hat\tau) \in \widehat{\mathsf{Ty}}$ if $\hat\tau \in \widehat{\mathsf{Ty}}$, where $\ell$ can be a variable or the concrete hook $\ell_{\bullet}$; (iv) Field types $\hat\tau_f \in \widehat{\mathsf{Ty}}$ if the parent struct $\hat\tau \in \widehat{\mathsf{Ty}}$ and field $f$ are visible. Crucially, $\widehat{\mathsf{Ty}}$ contains no free type variables ($\mathsf{TVar}$) and no unresolved projections ($\mathsf{Assoc}$ and $\mathsf{Field}$), but it retains lifetime variables to support schematic matching. All PCPN places will are strictly by $\widehat{\mathsf{Ty}}$.

We model resource availability using multisets. For a set $X$, $\mathbb{N}^{X}$ denotes the set of finite multisets $m: X \to \mathbb{N}$. We define standard operations: (i) Addition: $(m_1\oplus m_2)(x) := m_1(x)+m_2(x)$; (ii) Difference: $(m_1\ominus m_2)(x) := \max(m_1(x)-m_2(x),0)$; (iii) Inclusion: $m_1\subseteq m_2 \iff \forall x,\ m_1(x)\le m_2(x)$. These operations extend pointwise to mappings $M: P \to \mathbb{N}^{X}$.

\subsection{PCPNs}
\label{sec:method:pcpn}
We define a PCPN whose places are indexed by type and whose stack records outstanding references. Before defining the net structure, we fix the universes of identifiers and structure of solver evidences. Let $\mathbb{V}$ be a countably infinite set of value identifiers and $\mathbb{L}$ be a countably infinite set of region labels.
Let $\Theta$ be the set of finite type substitutions $\theta:\mathsf{TVar}\rightharpoonup\widehat{\mathsf{Ty}}$ and $\Lambda$ the set of finite lifetime valuations $\lambda:(\mathsf{LVar}\cup\{\ell_{\bullet}\})\rightharpoonup\mathbb{L}$.

A transition firing is justified by an \textit{instantiation record} $\iota \in \mathsf{Inst}$. This record captures how generic parameters are resolved and which fresh names are allocated. Formally, $\iota$ is a tuple: $\iota = \langle \theta, \lambda, \nu, \mu \rangle$ where $\theta \in \Theta$ and $\lambda \in \Lambda$ resolve the transition's static parameters, while $\nu: \mathsf{Vars} \to \mathbb{V}$ and $\mu: \mathsf{Vars} \to \mathbb{L}$ provide fresh identifiers and labels for the output tokens. The pushdown stack records outstanding borrows. We define the stack alphabet $\Gamma$ as: $\Gamma ~::=~ \mathsf{Freeze}(v)\mid \mathsf{Shr}(v_o,v_r,L)\mid \mathsf{Mut}(v_o,v_r,L)$ where $v, v_o, v_r \in \mathbb{V}$ and $L \in \mathbb{L}$. We write $\Gamma^*$ for the set of finite strings over $\Gamma$. A stack instance $S \in \Gamma^*$ implies a sequence of frames, and $\bar\gamma \in \Gamma^*$ denotes a segment of the stack pushed or popped atomically.

\begin{definition}[Pushdown Colored Petri Net]
\label{def:pcpn}
A PCPN is a tuple
\begin{equation*}
   \mathcal{N}_{\mathcal{C}}=(P,T,W,\mathsf{Col},\Gamma,M_0,\mathsf{Guard},\mathsf{Out},\mathsf{Act}), 
\end{equation*}
where: $P$ is a finite set of places indexed by capabilities and ground types: $P:=\{\,p^{\kappa}_{\hat\tau}\mid \kappa\in \{\mathsf{own},\mathsf{frz},\mathsf{blk}\},~\hat\tau\in\widehat{\mathsf{Ty}}\,\}$; $T$ is a finite set of transitions, comprising API-derived call rules and fixed structural rules; $W:(P\times T)\cup(T\times P)\to\mathbb{N}$ is the flow weight function; $\mathsf{Col} := \mathbb{V}\times\Theta\times\Lambda$ is the set of token colors. A token $c=\langle v,\theta,\lambda\rangle$ carries a value ID and its instantiation context; $M_0:P\to\mathbb{N}^{\mathsf{Col}}$ is the initial marking; $\mathsf{Guard}(t): (\langle M,S\rangle, \chi, \iota) \mapsto \mathsf{Prop}$ is a predicate checking if transition $t$ is enabled under configuration $\langle M,S\rangle$, input choice $\chi$, and instantiation record $\iota \in \mathsf{Inst}$; $\mathsf{Out}(t): (\langle M,S\rangle, \chi, \iota) \mapsto \pi$ defines the multiset of output tokens $\pi \in \mathbb{N}^{\mathsf{Col}}$ produced by the firing; and $\mathsf{Act}(t) \in \{\epsilon, \mathsf{push}(\bar\gamma), \mathsf{pop}(\bar\gamma)\}$ defines the stack operation, where $\bar\gamma \in \Gamma^*$ is a string of borrow frames.
\end{definition}
Let $M := P \to \mathbb{N}^{\mathsf{Col}}$ be the set of colored markings. For $p\in P$ and a multiset $m\in\mathbb{N}^{\mathsf{Col}}$, define the \emph{placed multiset}
$p\langle m\rangle \in M$ by
$(p\langle m\rangle)(p)=m$ and $(p\langle m\rangle)(p')=\emptyset$ for $p'\neq p$.
For a single token $c\in\mathsf{Col}$, write $p\langle c\rangle$ for $p\langle [c]\rangle$.

Although the set of places effectively triples the number of ground types, this linear expansion is deliberate. By encoding the borrowing states structurally into places rather than token colors, we enforce access control via net topology rather than complex guard predicates. This significantly simplifies the enabling check for transitions. Furthermore, since references are treated as first-class types in $\widehat{Ty}$, the model naturally supports nested borrows without additional rules.

\begin{definition}[Configurations]
\label{def:pcpn:config}
A \emph{configuration} of PCPN is a pair $\Cfg=\langle M,S\rangle$, where 
$M: P \to \mathbb{N}^{\mathsf{Col}}$ is a colored marking mapping each place to a multiset of tokens, and  $S \in \Gamma^{*}$ is a stack of borrow records with the top element at the leftmost position.
The initial configuration is $\Cfg_0 = \langle M_0, \epsilon \rangle$.
\end{definition}

\begin{definition}[Enabling]
\label{def:pcpn:enabling}
Let $\Cfg=\langle M,S\rangle$ be a configuration. A transition $t \in T$ is enabled at $\Cfg$ if there exist: an \textit{input choice} $\chi: P \to \mathbb{N}^{\mathsf{Col}}$ such that $\chi \subseteq M$ and $|\chi(p)| = W(p,t)$ for all $p \in P$; and an instantiation record $\iota \in \mathsf{Inst}$ providing variable valuations and fresh identifiers.

The pair $\varphi = \langle t, \chi, \iota \rangle$ constitutes an enabled \emph{firing instance} if the following conditions hold that Guard Satisfaction, i.e., $\mathsf{Guard}(t)(\Cfg, \chi, \iota)$ evaluates to true; and Stack Discipline, i.e., if $\mathsf{Act}(t) = \mathsf{pop}(\bar\gamma)$, then $S = \bar\gamma \cdot S_{rem}$ for some remainder $S_{rem} \in \Gamma^*$.
\end{definition}

\begin{definition}[Firing Rule]
\label{def:pcpn:firing}
An enabled firing instance $\varphi = \langle t, \chi, \iota \rangle$ transforms a configuration $\Cfg = \langle M, S \rangle$ into a new configuration $\Cfg' = \langle M', S' \rangle$, denoted as $\Cfg [\varphi\rangle \Cfg'$, according to the following rules:
\begin{itemize}
    \item[1.] Token Flow: $M' = (M \ominus \chi) \oplus \pi$, where $\pi = \mathsf{Out}(t)(\Cfg, \chi, \iota)$ is the multiset of produced tokens; and
    \item[2.] Stack Update: $S'$ is derived from $S$ by $\mathsf{Act}(t)$:
    \[
    S' = \begin{cases}
        S & \text{if } \mathsf{Act}(t) = \epsilon \\
        \bar\gamma \cdot S & \text{if } \mathsf{Act}(t) = \mathsf{push}(\bar\gamma) \\
        S_{rem} & \text{if } \mathsf{Act}(t) = \mathsf{pop}(\bar\gamma) \text{ and } S=\bar\gamma \cdot S_{rem}
    \end{cases}
    \]
\end{itemize}
A finite sequence of firing instances $\varphi_1 \dots \varphi_n$ is a \emph{valid trace} if there exist configurations $\Cfg_0, \dots, \Cfg_n$ such that $\Cfg_{i-1} [\varphi_i\rangle \Cfg_i$ for all $1 \le i \le n$.
\end{definition}
In standard CPNs, a transition is enabled simply by matching token colors to arc expressions (e.g., matching an integer value). However, in our setting, transitions represent \textit{polymorphic} Rust functions. A single transition schema does not describe a fixed transformation on ground values but rather a \textit{template} valid for infinite instantiations of \textit{Generic} \texttt{T}. Consequently, determining whether a transition is enabled is not merely a check of token availability but a constraint solving problem. The net must deduce the concrete types of generic parameters from the input tokens (Unification) and verify that these types satisfy the function's interface contracts (Obligation Discharge). The following section formalizes this process by refining the standard notion of binding elements via a structured constraint-satisfaction mechanism..

\subsection{Binding Elements and Guarded Enabling}
\label{sec:method:binding}
A subtle but critical aspect of our model is the connectivity between generic transitions and places indexed by \textit{ground types} (i.e., $\widehat{\mathsf{Ty}}$). Since places in $P$ are indexed by ground types $\widehat{\mathsf{Ty}}$, while a generic transition $t$ specifies input requirements as type schemes containing variables (e.g., $\alpha \in \mathsf{TVar}$), the arc weights $W(p, t)$ cannot be statically fixed for all generic instantiations. Instead, we interpret the connectivity \textit{symbolically}. The input choice $\chi$ in Def.~\ref{def:pcpn:enabling} represents a \textit{candidate connection}. The validity of connecting a place $p^{\kappa}_{\hat\tau}$ to a transition input port expecting type scheme $\tau_{in}$ is dynamically determined by the unification guard: $ \Sigma \vdash \tau_{in} \unify \hat\tau \Rightarrow (\sigma, \dots)$. Here, the ground type $\hat\tau$ serves as the evidence that constrains and resolves the type variable $\alpha$ in $\tau_{in}$. Thus, places do participate in binding: they act as the source of ground truth that collapses the generic template into a specific instance.

The compilation environment $\Sigma(\mathcal{C})$ provides the ground truth for checking validity. We verify constraints against two fact tables derived from the crate:
\begin{itemize}
    \item[1.] $\mathsf{ImplFact}_{\Sigma} \subseteq \widehat{\mathsf{Ty}} \times \mathsf{Trait}$: The set of known trait implementations. $(\hat\tau, \mathsf{Tr}) \in \mathsf{ImplFact}_{\Sigma}$ denotes that type $\hat\tau$ implements trait $\mathsf{Tr}$.
    \item[2.] $\mathsf{AssocFact}_{\Sigma} \subseteq \widehat{\mathsf{Ty}} \times \mathsf{Trait} \times \mathsf{Name} \times \widehat{\mathsf{Ty}}$: The table of associated types. $(\hat\tau, \mathsf{Tr}, A, \hat\tau_A) \in \mathsf{AssocFact}_{\Sigma}$ records that $\langle \hat\tau \text{ as } \mathsf{Tr} \rangle::A$ resolves to $\hat\tau_A$. 
\end{itemize}
Due to Rust's coherence rules, associated type resolution is functional. We denote the lookup as $\mathsf{AssocTy}_{\Sigma}(\hat\tau, \mathsf{Tr}, A) = \hat\tau_A$.

The guard of a transition generates a set of \textit{obligations} $\mathcal{O}$ that must be entailed by the environment. The syntax of an obligation atom $o \in \mathsf{Obl}$ is defined as:
\begin{equation*}
    o ~::=~ \tau : \mathsf{Tr} ~\mid~ \mathsf{Assoc}(\tau, \mathsf{Tr}, A) = \tau' ~\mid~ \ell_1 : \ell_2
\end{equation*}
Here, $\ell_1 : \ell_2$ represents the outlives relationship ($\ell_1 \sqsupseteq \ell_2$).

\subsubsection{Phase 1: Unification via Observation.}
Unification matches the \textit{expected type scheme} of a transition's input port against the \textit{observed ground type} of a selected token. This process produces a partial substitution $\sigma$ and may emit additional equality obligations $\mathcal{O}_u$ (e.g., when matching associated types). A \emph{substitution record} is a pair $\sigma=\langle\theta,\lambda\rangle$ where
$\theta:\mathsf{TVar}\rightharpoonup\widehat{\mathsf{Ty}}$ and
$\lambda:(\mathsf{LVar}\cup\{\ell_{\bullet}\})\rightharpoonup\mathbb{L}$ are finite partial maps. Two finite maps $m_1$ and $m_2$ are \emph{compatible}, written $m_1 \# m_2$, if
$\forall x\in\dom(m_1)\cap\dom(m_2).\, m_1(x)=m_2(x)$.
If compatible, we define their join $m_1\sqcup m_2$ as the union map.
Then we lift this pointwise to substitution records:
\begin{equation*}
    \langle\theta_1,\lambda_1\rangle \sqcup \langle\theta_2,\lambda_2\rangle
:= \langle \theta_1\sqcup\theta_2,\ \lambda_1\sqcup\lambda_2\rangle
\quad(\text{defined iff }\theta_1\#\theta_2\ \text{and }\lambda_1\#\lambda_2).
\end{equation*}
We extend join to pairs $(\sigma,\mathcal{O})$ by $(\sigma_1,\mathcal{O}_1)\bigsqcup(\sigma_2,\mathcal{O}_2) := (\sigma_1\sqcup\sigma_2,\ \mathcal{O}_1\cup\mathcal{O}_2)$, undefined if $\sigma_1$ and $\sigma_2$ are incompatible. For a finite family, $\bigsqcup_i(\sigma_i,\mathcal{O}_i)$ denotes iterated join.

First, we define the auxiliary function $\fvL(\tau)$ which extracts the set of lifetime labels occurring in type $\tau$. It is defined inductively:
\begin{align*}
\fvL(g)&=\emptyset &
\fvL(\alpha)&=\emptyset\\
\fvL(\tau_0\langle\tau_1,\dots,\tau_n\rangle)
&=\bigcup_{i=0}^n \fvL(\tau_i) &
\fvL((\tau_1,\dots,\tau_n))
&=\bigcup_i\fvL(\tau_i)\\
\fvL([\tau])&=\fvL(\tau) &
\fvL(\mathsf{Ref}^q_{\ell}(\tau))
&=\{\ell\}\cup\fvL(\tau)\\
\fvL(\mathsf{Assoc}(\tau,\mathsf{Tr},A))&=\fvL(\tau) &
\fvL(\mathsf{Field}(\tau,f))&=\fvL(\tau).
\end{align*}

Let $\TokObs(p, c)$ denote the observable structure of token $c$ at place $p$. For token $c = \langle v, \theta_c, \lambda_c \rangle$ residing in $p^{\kappa}_{\hat\tau}$, we observe the ground type $\hat\tau$ and the relevant lifetime labels: $\TokObs(p^{\kappa}_{\hat\tau}, c) ~:=~ \langle \hat\tau, \lambda_c \negthinspace\upharpoonright_{\fvL(\hat\tau)} \rangle$.

We define the unification judgment $\Sigma \vdash \tau \unify \langle \hat\tau,\lambda\rangle \Rightarrow (\sigma, \mathcal{O}_u)$ by using the following inference rules:
\begin{figure*}[ht]
\centering
\small
\setlength{\lineskip}{2pt}
\renewcommand{\arraystretch}{1.0}
\mprset{sep=4pt, andskip=1.5em}

\begin{mathpar}
    \inferrule[\textsc{U-GTy}]
    { }
    { \Sigma \vdash g \unify \langle g, \lambda \rangle \Rightarrow (\langle\emptyset,\emptyset\rangle, \emptyset) }
    \and
    \inferrule[\textsc{U-TVar}]
    { }
    { \Sigma \vdash \alpha \unify \langle \hat\tau, \lambda \rangle \Rightarrow (\langle[\alpha \mapsto \hat\tau],\emptyset\rangle, \emptyset) }
    \and
    \inferrule[\textsc{U-Slice}]
    { \Sigma \vdash \tau \unify \langle \hat\tau,\lambda\rangle \Rightarrow (\sigma,\mathcal{O}) }
    { \Sigma \vdash [\tau] \unify \langle [\hat\tau],\lambda\rangle \Rightarrow (\sigma,\mathcal{O}) }

    \inferrule[\textsc{U-Field}]
    { }
    { \Sigma \vdash \mathsf{Field}(\tau,f) \unify \langle \hat\tau, \lambda \rangle
      \Rightarrow (\langle\emptyset,\emptyset\rangle,\{ \mathsf{Field}(\tau,f)=\hat\tau \}) }
    \and
    \inferrule[\textsc{U-Assoc}]
    { }
    { \Sigma \vdash \mathsf{Assoc}(\tau, \mathsf{Tr}, A) \unify \langle \hat\tau, \lambda \rangle
      \Rightarrow (\langle\emptyset,\emptyset\rangle,\{ \mathsf{Assoc}(\tau, \mathsf{Tr}, A) = \hat\tau \}) }

    \inferrule[\textsc{U-Ref}]
    { \Sigma \vdash \tau \unify \langle \hat\tau, \lambda \rangle \Rightarrow (\sigma_0, \mathcal{O}_0) \\
      \lambda(\ell') = L }
    { \Sigma \vdash \mathsf{Ref}^q_{\ell}(\tau) \unify \langle \mathsf{Ref}^q_{\ell'}(\hat\tau), \lambda \rangle
      \Rightarrow (\sigma_0 \sqcup \langle\emptyset,[\ell\mapsto L]\rangle, \mathcal{O}_0) }

    \inferrule[\textsc{U-Tuple}]
    { \forall i.\ \Sigma \vdash \tau_i \unify \langle \hat\tau_i,\lambda\rangle \Rightarrow (\sigma_i,\mathcal{O}_i) \\
      (\sigma,\mathcal{O})=\bigsqcup_i(\sigma_i,\mathcal{O}_i)}
    { \Sigma \vdash (\tau_1,\dots,\tau_n) \unify \langle (\hat\tau_1,\dots,\hat\tau_n),\lambda\rangle \Rightarrow (\sigma,\mathcal{O}) }
    \and
    \inferrule[\textsc{U-App}]
    { \Sigma \vdash \tau_0 \unify \langle \hat\tau_0,\lambda\rangle \Rightarrow (\sigma_0,\mathcal{O}_0) \\
      \forall i.\ \Sigma \vdash \tau_i \unify \langle \hat\tau_i,\lambda\rangle \Rightarrow (\sigma_i,\mathcal{O}_i) \\
      (\sigma,\mathcal{O})=\bigsqcup_{i=0}^n (\sigma_i,\mathcal{O}_i) }
    { \Sigma \vdash \tau_0\langle\tau_1,\dots,\tau_n\rangle \unify
      \langle \hat\tau_0\langle\hat\tau_1,\dots,\hat\tau_n\rangle,\lambda\rangle
      \Rightarrow (\sigma,\mathcal{O}) }
\end{mathpar}
\end{figure*}

\subsubsection{Phase 2: Entailment and Resolution.}
An instantiated obligation set $\mathcal{O}$ is satisfied if every atom is derivable from $\Sigma$ and the current stack $S$. We formalize this by using entailment rules:

\begin{mathpar}
\inferrule[\textsc{E-Trait}]
  { (\theta(\tau), \mathsf{Tr}) \in \mathsf{ImplFact}_{\Sigma} }
  { \Sigma; S \vdash \sigma(\tau : \mathsf{Tr}) }

\inferrule[\textsc{E-Life}]
  { S \vdash \lambda(\ell_1) \sqsupseteq \lambda(\ell_2) }
  { \Sigma; S \vdash \sigma(\ell_1 : \ell_2) }

\inferrule[\textsc{E-Assoc}]
  { \mathsf{AssocTy}_{\Sigma}(\theta(\tau), \mathsf{Tr}, A) = \theta(\tau') }
  { \Sigma; S \vdash \sigma(\mathsf{Assoc}(\tau, \mathsf{Tr}, A) = \tau') }
\end{mathpar}
The outlives relation $S \vdash L_1 \sqsupseteq L_2$ reflects the LIFO validity: it holds if $L_1 = L_2$ or if $L_1$ appears deeper in the stack $S$ than $L_2$ does.

\subsubsection{Phase 3: Unified Enabling Judgment.}
We can now formally define the $\SolveGuard$ predicate. A firing instance is valid if we can find a \textit{complete} instantiation record $\iota$ that extends the partial unification result and satisfies all obligations.

Let $\Type_{in}(t, p)$ be the type scheme expected by transition $t$ at input place $p$. To accommodate the complexity of the enabling condition within limited width, we present the rule by using a stacked structure:
\begin{equation*}
\label{eq:gsolve}
\frac{
  \begin{array}{@{}l}
    (1)\ \textbf{Input Matching:} \\
    \qquad \forall p \in {}^{\bullet}t, \forall c \in \chi(p):\ 
    \Sigma \vdash \Type_{in}(t, p) \unify \TokObs(p, c) \Rightarrow (\sigma_{p,c}, \mathcal{O}_{p,c}) \\[5pt]
    
    (2)\ \textbf{Consistency:} \\
    \qquad (\sigma_{forced}, \mathcal{O}_{u}) = \bigsqcup_{p,c} (\sigma_{p,c}, \mathcal{O}_{p,c}) \\[5pt]
    
    (3)\ \textbf{Completion:} \\
    \qquad \sigma^\star \succeq \sigma_{forced} \text{ is a complete binding for } \mathsf{Vars}(t) \\[5pt]
    
    (4)\ \textbf{Freshness:} \\
    \qquad \FreshOK(\langle M, S \rangle, \sigma^\star, \nu, \mu) \\[5pt]
    
    (5)\ \textbf{Entailment:} \\
    \qquad \Sigma; S \vdash \sigma^\star \models (\mathcal{O}(t) \cup \mathcal{O}_{u})
  \end{array}
}{
  \Sigma; \langle M, S \rangle \vdash_{\text{enable}} t, \chi \leadsto \langle \sigma^\star, \nu, \mu \rangle
} \ (\textsc{G-Solve})
\end{equation*}
The judgment proceeds in two logical stages: binding synthesis and validity checking.
First, the consistency step merges partial unifications from all input tokens into a forced substitution $\sigma_{forced}$ and a collected obligation set $\mathcal{O}_{u}$. The join operation $\bigsqcup$ fails if any token implies a conflicting assignment for the same variable. Since $\sigma_{forced}$ may be partial (e.g., return-type-only parameters remain unconstrained), the completion step non-deterministically selects a total extension $\sigma^\star \succeq \sigma_{forced}$. This extension must map all variables in $\mathsf{Vars}(t)$—the set of generic type and lifetime parameters declared by $t$—to valid ground types in $\widehat{\mathsf{Ty}}$ and region labels in $\mathbb{L}$. 
Second, the judgment enforces structural and semantic integrity. The freshness predicate $\FreshOK(\langle M, S \rangle, \sigma^\star, \nu, \mu)$ ensures that the new value identifiers $\nu$ and region labels $\mu$ generated for output tokens are disjoint from those currently live in marking $M$ and stack $S$. Finally, the entailment check requires that the total obligation set—combining static requirements $\mathcal{O}(t)$ from the signature and dynamic requirements $\mathcal{O}_u$ emitted by unification—is fully derivable from the environment $\Sigma$ and the current borrow stack $S$.

\section{Constructing PCPN from $\mathcal{C}$}
\label{sec:method:construct}

The construction of PCPN $\mathcal{N}_{\mathcal{C}}$ transforms the static signature environment $\Sigma(\mathcal{C})$ into a dynamic net structure. This involves two steps: monomorphizing API calls into concrete transitions, and instantiating a fixed set of structural schemas.

\subsection{API-Derived Call Transitions}
\label{sec:method:construct:api}

From $\Sigma(\mathcal{C})$, we extract a finite set of callable items $\mathcal{F}$. 
While $\mathsf{ImplFact}_{\Sigma}$ tells us which types implement which \textit{traits}, it does not create the net structure. We must explicitly generate a transition for each valid concrete instantiation of a generic function.

\begin{definition}[Call Transition Instantiation]
\label{def:call-trans}
Let $\Theta_f := \{\,\theta:\bar\alpha_f\to \widehat{\mathsf{Ty}}\,\}$ be the set of all substitutions mapping parameters of $f$ to ground types in $\widehat{\mathsf{Ty}}$.
For each $f\in\mathcal{F}$, we generate a transition $t_{f,\theta}$ for every substitution $\theta\in\Theta_f$ satisfying the following conditions: All argument and return types resolve to valid types in $\widehat{\mathsf{Ty}}$; and all static obligations $\mathcal{O}_0$ (derived from the signature of $f$) are statically entailed by $\mathsf{ImplFact}_{\Sigma}$. 
The constructed $t_{f,\theta}$ has input/output places determined by the resolved types and carries the remaining obligations in its guard.
\end{definition}

\subsection{Structural Transition Schemas}
\label{sec:method:construct:structural}

We instantiate a fixed family of structural schemas to model Rust's ownership and borrowing rules.
To ensure formal rigor, we define the auxiliary constructors used in the schemas.
Let $c = \langle v, \theta, \lambda \rangle$ be an input token. $\mathsf{NewVal}(c, v') := \langle v', \theta, \lambda \rangle$ creates a new value with the same type/lifetime context but a fresh ID. $\mathsf{NewRef}(c, v', L) := \langle v', \theta, \lambda[\ell_{\bullet} \mapsto L] \rangle$ creates a reference token with a fresh ID and a specific region label $L$.

In Tables~\ref{tab:schemas:own}--\ref{tab:schemas:deref} , $\chi(p)$ denotes the single token consumed from place $p$, and $\pi$ denotes the multiset of produced tokens. All schemas implicitly require $\FreshOK$ if they generate new IDs ($\nu, \mu$).

The transition schemas in Table~\ref{tab:schemas:own} formalize Rust's affine and linear typing discipline by defining the consumption and production of owned value tokens. Rule 1, denoted as $\mathsf{Move}_{\hat\tau}$, enforces strict linearity for non-\textsf{Copy} types. It consumes a token from $p^{\mathsf{own}}_{\hat\tau}$ and produces an empty post-multiset, effectively removing the value from the current scope to model a move operation or drop. Conversely, Rule 2, $\mathsf{CopyUse}_{\hat\tau}$, applies to types implementing the \textsf{Copy} trait. It consumes the token but immediately regenerates it in the same place, treating the value as a persistent resource that can be duplicated implicitly. The explicit creation of independent values is handled by Rules 2$^\star$ and 2$^\prime$. $\mathsf{DupCopy}$ performs a bitwise copy, while $\mathsf{DupClone}$ invokes a deep copy via the \textsf{Clone} trait. Crucially, these are the only structural rules that increase the net cardinality of tokens in the system. Therefore, both the finiteness of the state space and the termination of the reachability analysis rely primarily on restricting the firing frequency of these specific duplication schemas within the budget.

\begin{table}
\centering
\scriptsize
\setlength{\tabcolsep}{4pt}
\renewcommand{\arraystretch}{1.3}
\caption{Ownership and Duplication Schemas.}
\label{tab:schemas:own}
\begin{tabular}{l l l l}
\hline
\textbf{Rule} & \textbf{Schema Name} & \textbf{Output Multiset ($\pi$)} & \textbf{Guard Condition} \\
\hline
1 & $\mathsf{Move}_{\hat\tau}$ & $\varnothing$ & $\neg \mathsf{IsCopy}(\hat\tau)$ \\
2 & $\mathsf{CopyUse}_{\hat\tau}$ & $p^{\mathsf{own}}_{\hat\tau}\langle c\rangle$ & $\mathsf{IsCopy}(\hat\tau)$ \\
2$^\star$ & $\mathsf{DupCopy}_{\hat\tau}$ & $p^{\mathsf{own}}_{\hat\tau}\langle c\rangle \oplus p^{\mathsf{own}}_{\hat\tau}\langle c'\rangle$ & $\mathsf{IsCopy}(\hat\tau) \wedge c'=\mathsf{NewVal}(c,\nu)$ \\
2$^{\prime}$ & $\mathsf{DupClone}_{\hat\tau}$ & $p^{\mathsf{own}}_{\hat\tau}\langle c\rangle \oplus p^{\mathsf{own}}_{\hat\tau}\langle c'\rangle$ & $\hat\tau:\mathsf{Clone} \wedge c'=\mathsf{NewVal}(c,\nu)$ \\
3$^{\prime}$ & $\mathsf{DropOwn}_{\hat\tau}$ & $\varnothing$ & $\top$ \\
\hline
\end{tabular}
\end{table}

Table~\ref{tab:schemas:borrow} implements the \textit{Readers-Writer Lock} mechanism through the coordination of state places ($p^{\mathsf{frz}}, p^{\mathsf{blk}}$) and pushdown stack. The creation of shared references involves a two-phase freeze logic. Rule 4 initiates the first borrow by moving the owner token from $p^{\mathsf{own}}$ to $p^{\mathsf{frz}}$ and pushing a composite stack frame $\mathsf{Shr}(\dots) \cdot \mathsf{Freeze}(v_o)$. This $\mathsf{Freeze}(v_o)$ marker acts as a sentinel, indicating that this specific transition is responsible for suspending ownership. Subsequent aliasing borrows, defined in Rule 5, merely cycle the token within $p^{\mathsf{frz}}$ and push an $\mathsf{Shr}$ frame. The discharge process enforces LIFO validity: popping an $\mathsf{Shr}$ frame (Rule 8) leaves the owner frozen to support remaining readers, whereas popping the frame containing the sentinel (Rule 9) triggers the unfreeze operation, restoring the token to $p^{\mathsf{own}}$. For mutable borrows, Rule 6 enforces exclusivity by moving the owner to $p^{\mathsf{blk}}$ and pushing a $\mathsf{Mut}$ frame. The owner remains inaccessible until the corresponding $\mathsf{Mut}$ frame is popped in Rule 7, guaranteeing that no aliasing occurs during the mutation scope. We denote $p^{\mathsf{ref}}_{\hat\tau} := p^{\mathsf{own}}_{\mathsf{Ref}^{\mathsf{shr}}_{\ell_{\bullet}}(\hat\tau)}$ and $p^{\mathsf{mut}}_{\hat\tau} := p^{\mathsf{own}}_{\mathsf{Ref}^{\mathsf{mut}}_{\ell_{\bullet}}(\hat\tau)}$ for brevity in the table, but they refer to the standard places defined in Def.~\ref{def:pcpn}.

\begin{table}
\centering
\scriptsize
\setlength{\tabcolsep}{3pt}
\renewcommand{\arraystretch}{1.3}
\caption{Borrow Creation and Discharge Schemas}
\label{tab:schemas:borrow}
\begin{tabular}{l l l l}
\hline
\textbf{No.} & \textbf{Transition Flow} & \textbf{Stack Action ($\mathsf{Act}$)} & \textbf{Output Token Construction} \\
\hline
4 & $p^{\mathsf{own}}_{\hat\tau} \to p^{\mathsf{frz}}_{\hat\tau} \oplus p^{\mathsf{ref}}_{\hat\tau}$ & $\mathsf{push}\ \mathsf{Shr}(v_o,v_r,L) \cdot \mathsf{Freeze}(v_o)$ & $c_r = \mathsf{NewRef}(c, \nu, \mu)$ \\
5 & $p^{\mathsf{frz}}_{\hat\tau} \to p^{\mathsf{frz}}_{\hat\tau} \oplus p^{\mathsf{ref}}_{\hat\tau}$ & $\mathsf{push}\ \mathsf{Shr}(v_o,v_r,L)$ & $c_r = \mathsf{NewRef}(c, \nu, \mu)$ \\
6 & $p^{\mathsf{own}}_{\hat\tau} \to p^{\mathsf{blk}}_{\hat\tau} \oplus p^{\mathsf{mut}}_{\hat\tau}$ & $\mathsf{push}\ \mathsf{Mut}(v_o,v_r,L)$ & $c_r = \mathsf{NewRef}(c, \nu, \mu)$ \\
\hline
7 & $p^{\mathsf{blk}}_{\hat\tau} \oplus p^{\mathsf{mut}}_{\hat\tau} \to p^{\mathsf{own}}_{\hat\tau}$ & $\mathsf{pop}\ \mathsf{Mut}(v_o,v_r,L)$ & -- \\
8 & $p^{\mathsf{frz}}_{\hat\tau} \oplus p^{\mathsf{ref}}_{\hat\tau} \to p^{\mathsf{frz}}_{\hat\tau}$ & $\mathsf{pop}\ \mathsf{Shr}(v_o,v_r,L)$ & -- \\
9 & $p^{\mathsf{frz}}_{\hat\tau} \oplus p^{\mathsf{ref}}_{\hat\tau} \to p^{\mathsf{own}}_{\hat\tau}$ & $\mathsf{pop}\ \mathsf{Shr}(v_o,v_r,L) \cdot \mathsf{Freeze}(v_o)$ & -- \\
\hline
\end{tabular}
\end{table}

The structural decomposition of types is modeled in Table~\ref{tab:schemas:proj}, where ownership and borrowing rules are propagated from parent structures to their fields. Rule 10 formalizes a partial move by consuming the parent struct token and producing a token for the specific field $\tau_f$. For shared references, Rule 11 derives a field reference $\&\tau_f$ from a parent reference $\&\tau$ without modifying the stack, as the parent's frozen state implicitly covers all fields. The most complex case is the \textit{splitting borrow} formalized in Rule 12. When deriving a mutable reference to a field ($\& \mathsf{mut} \tau_f$) from a mutable parent reference ($\& \mathsf{mut} \tau$), the system treats the parent reference itself as the resource owner. The parent reference token is moved to a blocked state $p^{\mathsf{blk}}_{\mathsf{Ref}}$ to prevent simultaneous access to the rest of the struct, and a new $\mathsf{Mut}$ frame is pushed. This mechanism ensures that disjoint fields can be mutated safely only when the borrow checker's granularity allows, or prevents access when the parent is already uniquely borrowed.

\begin{table}
\centering
\scriptsize
\setlength{\tabcolsep}{3pt}
\renewcommand{\arraystretch}{1.3}
\caption{Field Projection Schemas (for field $f$ of type $\tau$)}
\label{tab:schemas:proj}
\begin{tabular}{l l l l}
\hline
\textbf{No.} & \textbf{Transition Flow} & \textbf{Stack Action} & \textbf{Token/Guard} \\
\hline
10 & $p^{\mathsf{own}}_{\tau} \to p^{\mathsf{own}}_{\tau_f}$ & $\epsilon$ & $c_f = \mathsf{NewVal}(c, \nu)$ \\
11 & $p^{\mathsf{ref}}_{\tau} \to p^{\mathsf{ref}}_{\tau} \oplus p^{\mathsf{ref}}_{\tau_f}$ & $\epsilon$ & $c_c = \mathsf{NewVal}(c, \nu)$ \\
12 & $p^{\mathsf{mut}}_{\tau} \to p^{\mathsf{blk}}_{\mathsf{Ref}^{\mathsf{mut}}(\tau)} \oplus p^{\mathsf{mut}}_{\tau_f}$ & $\mathsf{push}\ \mathsf{Mut}(\dots)$ & $c_c = \mathsf{NewRef}(c, \nu, \mu)$ \\
13 & $p^{\mathsf{blk}}_{\mathsf{Ref}^{\mathsf{mut}}(\tau)} \oplus p^{\mathsf{mut}}_{\tau_f} \to p^{\mathsf{mut}}_{\tau}$ & $\mathsf{pop}\ \mathsf{Mut}(\dots)$ & (Stack match) \\
\hline
\end{tabular}
\end{table}

Table~\ref{tab:schemas:deref} defines operations performed through references, focusing on the mechanism of reborrowing essential for nested calls. Rule 14 allows reading a value from a shared reference only if the referent type is \textsf{Copy}, producing a new value token while maintaining the reference's liveness. Rules 15 and 16 formalize explicit reborrowing, where a mutable reference is temporarily borrowed out to create a shorter-lived sub-reference either mutable or shared. In this process, the original reference token is moved to $p^{\mathsf{blk}}_{\mathsf{Ref}}$, and a new reference token is generated with a fresh ID. The stack tracks the lifetime of this sub-borrow. When the corresponding frame is popped (Rules 15$'$ and 16$'$), the original reference is unblocked and returned to $p^{\mathsf{own}}_{\mathsf{Ref}}$.

\begin{table}
\centering
\scriptsize
\setlength{\tabcolsep}{3pt}
\renewcommand{\arraystretch}{1.3}
\caption{Dereference and Reborrowing Schemas}
\label{tab:schemas:deref}
\begin{tabular}{l l l l}
\hline
\textbf{No.} & \textbf{Transition Flow} & \textbf{Stack Action} & \textbf{Token/Guard} \\
\hline
14 & $p^{\mathsf{ref}}_{\tau} \to p^{\mathsf{ref}}_{\tau} \oplus p^{\mathsf{own}}_{\tau}$ & $\epsilon$ & $\mathsf{IsCopy}(\tau) \wedge c' = \mathsf{NewVal}(c, \nu)$ \\
15 & $p^{\mathsf{mut}}_{\tau} \to p^{\mathsf{blk}}_{\mathsf{Ref}^{\mathsf{mut}}(\tau)} \oplus p^{\mathsf{mut}}_{\tau}$ & $\mathsf{push}\ \mathsf{Mut}(\dots)$ & $c' = \mathsf{NewRef}(c, \nu, \mu)$ \\
15' & $p^{\mathsf{blk}}_{\mathsf{Ref}^{\mathsf{mut}}(\tau)} \oplus p^{\mathsf{mut}}_{\tau} \to p^{\mathsf{mut}}_{\tau}$ & $\mathsf{pop}\ \mathsf{Mut}(\dots)$ & (Stack match) \\
16 & $p^{\mathsf{mut}}_{\tau} \to p^{\mathsf{blk}}_{\mathsf{Ref}^{\mathsf{mut}}(\tau)} \oplus p^{\mathsf{ref}}_{\tau}$ & $\mathsf{push}\ \mathsf{Mut}(\dots)$ & $c' = \mathsf{NewRef}(c, \nu, \mu)$ \\
16' & $p^{\mathsf{blk}}_{\mathsf{Ref}^{\mathsf{mut}}(\tau)} \oplus p^{\mathsf{ref}}_{\tau} \to p^{\mathsf{mut}}_{\tau}$ & $\mathsf{pop}\ \mathsf{Mut}(\dots)$ & (Stack match) \\
\hline
\end{tabular}
\end{table}

\subsection{Correctness and Bisimulation}
\label{sec:method:correctness}
We establish the correctness of our synthesis approach by proving a strong bisimulation between PCPN and Symbolic Rust Semantics (SRS) defined in Section~\ref{sec_background_core}. This verification relies on the alignment between the open-ended nature of the Rust language and the finite structure of the constructed net.

While the syntax of SRS allows for arbitrary Rust programs, PCPN is constructed specifically to explore the state space reachable via the finite API $\Sigma(\mathcal{C})$. To make verification well-posed, we must restrict the domain of SRS to the finite universe established during net construction.

Recall that $\widehat{\mathsf{Ty}}$ is the finite set of ground types instantiated from the generic signatures of $\mathcal{C}$. We impose the Closed World Assumption on SRS:
\begin{enumerate}
    \item The program $\textsf{prog}$ may only declare variables and operate on resources having types $\tau \in \widehat{\mathsf{Ty}}$; and
    \item We model primitive literals (e.g., integers) as explicit 0-ary constructors in the callable set $\mathcal{F}$. This ensures that the creation of primitive values in SRS corresponds to explicit firing transitions in PCPN.
\end{enumerate}
Under the above assumptions, SRS and PCPN operate over the same finite universe of types.

We relate SRS and PCPN transitions through the set of labels $\mathsf{Lab}$, which includes API calls $f@\sigma$ and structural schema names. Since both systems allocate fresh identifiers dynamically, we compare configurations modulo $\beta$-renaming, a bijective renaming of value IDs and region labels, denoted $\Cfg_1 \equiv_\beta \Cfg_2$. We formalize the components of an SRS rule. For a label $l \in \mathsf{Lab}$ and an instantiation $\iota$, we denote: $\mathsf{PreCond}_\Sigma(l, \Cfg, \iota)$: The logical conjunction of all preconditions in the SRS rule for $l$ including resource presence, type equality, trait bounds, and freshness. $\mathsf{Consume}_l(\Cfg, \iota)$ removes a multiset of resources (tokens) from the configuration by rule $l$. $\mathsf{Produce}_l(\Cfg, \iota)$ adds a multiset of resources and  pushes/pops the stack frames by rule $l$. We write $\mathsf{SRS\text{-}post}(\Cfg, l, \iota)$ for the configuration resulting from applying these updates.

\begin{lemma}[Structural Transformer Coincidence]
\label{lem:struct-coincide}
For every structural label $l\in\mathsf{Lab}$ (corresponding to schemas in Tables~\ref{tab:schemas:own}--\ref{tab:schemas:deref}), let $t_l$ be the unique corresponding transition in the PCPN.
For any configuration $\Cfg$ and instantiation record $\iota$, the following hold: (1) Consumption Match, i.e., the input arc weights of $t_l$ exactly match the SRS resource consumption: $\mathsf{Consume}_l(\Cfg,\iota) = \chi \iff \chi \text{ matches } W^-(t_l)$; and (2) Production Match, i.e., the state update effect is identical modulo $\beta$-renaming. If enabled, $\mathsf{SRS\text{-}post}(\Cfg,l,\iota) \ \equiv_\beta \ \Fire(\Cfg, \langle t_l, \chi, \iota \rangle)$.
\end{lemma}

\begin{proof}
This follows directly from the construction of the PCPN schemas in Section~\ref{sec:method:construct}. Each SRS rule explicitly lists the resources required (e.g., $p^{\mathsf{own}}$). The corresponding PCPN schema (e.g., Rule 4) is constructed with exactly these places as input arcs with unit weight. Thus, $\chi$ forms a valid multiset for $t_l$ iff it satisfies the SRS consumption requirement. The SRS post-state is defined by $(M \ominus \chi) \oplus \pi$ and a stack update. The PCPN firing rule applies $\mathsf{Out}(t_l)$ and $\mathsf{Act}(t_l)$. By inspection of Tables~\ref{tab:schemas:own}--\ref{tab:schemas:deref}, for every rule, $\mathsf{Out}(t_l)$ generates the same multiset $\pi$ using the same auxiliary constructors $\mathsf{NewVal}/\mathsf{NewRef}$, and $\mathsf{Act}(t_l)$ performs the same push/pop operation. Any difference lies solely in the choice of fresh identifiers in $\iota$, which is absorbed by $\equiv_\beta$. \hfill $\blacksquare$
\end{proof}

\begin{lemma}[Guard Equivalence]
\label{lem:guard-eq}
For every label $l\in\mathsf{Lab}$ and configuration $\Cfg$, the SRS preconditions hold if and only if the PCPN guard solver succeeds. Let $\chi$ be a valid resource choice. Then:
\[
\mathsf{PreCond}_\Sigma(l,\Cfg,\iota)
\quad\Longleftrightarrow\quad
\iota \in \SolveGuard(t_l, \Cfg, \chi).
\]
\end{lemma}
\begin{proof}
The proof proceeds by showing that the three phases of $\SolveGuard$ are structurally isomorphic to the preconditions of the SRS rules.

\textbf{($\Rightarrow$ Completeness)} Assume $\mathsf{PreCond}_\Sigma$ holds. SRS requires that the types of selected resources match the rule's input signature (e.g., $c \in p^{\mathsf{own}}_{\hat\tau}$). This implies that the observed token type $\hat\tau_{obs}$ unifies with the schema's input type $\tau_{in}$. Thus, Phase 1 of $\SolveGuard$ succeeds, yielding substitution $\sigma_{unify}$. SRS validates trait bounds and lifetime relations (e.g., via $\CapOK$). These static checks are exactly the definition of the entailment relation $\Sigma; S \vdash \mathcal{O}$. Since SRS implies that these hold, Phase 2 of $\SolveGuard$ (Entailment) succeeds. SRS requires $\FreshOK(\Cfg, \nu, \mu)$. This is explicitly checked in Phase 3 of $\SolveGuard$. Since all phases succeed, $\iota \in \SolveGuard$.

\textbf{($\Leftarrow$ Soundness)} Assume $\iota \in \SolveGuard$.
This implies $\Sigma \vdash \text{Unify}$ (Phase 1), $\Sigma \vdash \text{Entail}$ (Phase 2), and $\FreshOK$ (Phase 3). The construction of $t_l$ ensures that its guard obligations $\mathcal{O}(t)$ are exactly the trait bounds from the SRS signature, and $\mathcal{O}_{u}$ captures the geometric constraints of the rule. Therefore, the satisfaction of $\SolveGuard$ implies that all logical predicates in $\mathsf{PreCond}_\Sigma$ are satisfied.  \hfill $\blacksquare$
\end{proof}

\begin{lemma}[$\beta$-Invariance]
\label{lem:beta-inv}
Let $\Cfg_1 \equiv_\beta \Cfg_2$. For every firing step $\Cfg_1 [\varphi_1\rangle \Cfg'_1$, there exists a firing step $\Cfg_2 [\varphi_2\rangle \Cfg'_2$ such that $\mathrm{lab}(\varphi_1) = \mathrm{lab}(\varphi_2)$ and $\Cfg'_1 \equiv_\beta \Cfg'_2$.
\end{lemma}

\begin{proof}
The proof follows from the structural definition of the transition schemas. Since all guard conditions (including equality checks, unification, and stack patterns) and state updates (multiset operations and stack actions) are defined parametrically over identifiers, they commute with any bijective renaming map. Thus, $\beta$-equivalence is a bisimulation on the configuration space. \hfill $\blacksquare$
\end{proof}

\begin{theorem}[Step Correspondence]
\label{thm:step-corr}
Let $\Cfg_s$ be an SRS configuration and $\Cfg_p$ be a PCPN configuration such that $\Cfg_s \equiv_\beta \Cfg_p$.
For any label $l \in \mathsf{Lab}$:
\begin{equation*}
    \Sigma(\mathcal{C}) \vdash \Cfg_s \xrightarrow{l} \Cfg'_s
    \quad\Longleftrightarrow\quad
    \exists \varphi.\ \Cfg_p[\varphi\rangle \Cfg'_p \ \wedge\ \mathrm{lab}(\varphi)=l \ \wedge\ \Cfg'_s \equiv_\beta \Cfg'_p.
\end{equation*}
\end{theorem}
\begin{proof}
We prove the equivalence by mutual implication, relying on the structural coincidence (Lemma~\ref{lem:struct-coincide}) and guard equivalence (Lemma~\ref{lem:guard-eq}).

For the forward direction, assume that SRS takes a step $\Sigma \vdash \Cfg_s \xrightarrow{l} \Cfg'_s$. By the Closed World Assumption, the involved types reside within $\widehat{\mathsf{Ty}}$, ensuring that the corresponding transition $t_l$ exists in the net. Since $\Cfg_s \equiv_\beta \Cfg_p$, we can identify a resource set $\chi$ in $\Cfg_p$ corresponding to the resources consumed by SRS. The validity of the SRS step implies that all preconditions holds; thus by Lemma~\ref{lem:guard-eq}, the PCPN guard is satisfied, meaning that $t_l$ is enabled. Furthermore, Lemma~\ref{lem:struct-coincide} guarantees that the state update performed by the SRS rule is functionally identical to the effect of firing $t_l$, producing a post-state $\Cfg'_p$ that is $\beta$-equivalent to $\Cfg'_s$.

Conversely, for the backward direction, assume the PCPN fires $\Cfg_p[\varphi\rangle \Cfg'_p$ with label $l$. Since the transition $t_l$ is enabled, the guard constraints are satisfied. By Lemma~\ref{lem:guard-eq}, this implies that the logical preconditions for the corresponding SRS statement are met. We can thus construct a valid SRS derivation using the same resources. By Lemma~\ref{lem:struct-coincide}, the deterministic update logic of the SRS mirrors the net's transition effect, resulting in a configuration $\Cfg'_s$ that is $\beta$-equivalent to $\Cfg'_p$. \hfill $\blacksquare$
\end{proof}

\begin{theorem}[Strong Bisimulation]
\label{thm:bisim}
The relation $\mathcal{R} = \{ (\Cfg_s, \Cfg_p) \mid \Cfg_s \equiv_\beta \Cfg_p \}$ is a strong bisimulation between $\Sigma$-restricted SRS Label Transition System (LTS) and PCPN firing LTS.
\end{theorem}
\begin{proof}
The result is a direct consequence of Theorem~\ref{thm:step-corr} and $\beta$-invariance property (Lemma~\ref{lem:beta-inv}). Since $\mathcal{R}$ is defined by $\beta$-equivalence, and Theorem~\ref{thm:step-corr} establishes that every step from $\beta$-equivalent states leads to $\beta$-equivalent states with the same label, the simulation conditions are strictly satisfied in both directions. This implies that any trace realizable in the verified PCPN corresponds to a valid SRS program execution, and conversely, any valid execution of a program constructed from the library is explorable by the net. \hfill $\blacksquare$
\end{proof}

\section{Reachability Analysis and Snippet Synthesis}
\label{sec:analysis}
This section develops a finite, bounded reachability graph on top of $\mathcal{N}_{\mathcal{C}}$ and uses it as a search space for witness extraction and snippet emission~\cite{solar2006combinatorial}. Building on Theorem~\ref{thm:bisim}, which guarantees semantic adequacy of PCPN with respect to SRS, we focus on the algorithmic pipeline that turns the net into executable Rust witnesses. Starting from an extracted signature environment $\Sigma(\mathcal{C})$ and callable set $\mathcal{F}$, we construct $\mathcal{N}_{\mathcal{C}}$ as described in Section~\ref{sec:method}. We then fix explicit bounds on type instantiations, token multiplicities, and borrow-stack depth. Configurations are quotiented up to $\beta$-renaming by using a canonicalization function $\Canon$ to ensure finiteness, and a finite reachability graph enabled firings \cite{gabbay2002new},\cite{pitts2013nominal}. Synthesis tasks are performed by searching this graph with an objective predicate, optionally guided by strategy heuristics to prioritize effective parameter or type bindings and suppress endlessly applicable value producers. Once a goal node is reached, we backtrack its firing sequence and emit a Rust snippet by translating each transition into a corresponding core statement.

\subsection{Explicit Bounds and Duplication Limits}
\label{sec:analysis:bounds}
To ensure termination, we must constrain the infinite potential of data duplication and recursion. We define a bound tuple $\mathcal{B} = (\widehat{\mathsf{Ty}}, B, D)$, where $\widehat{\mathsf{Ty}}$ is the finite universe of instantiated types, $B: P \to \mathbb{N}$ limits token count per place, and $D \in \mathbb{N}$ limits stack depth.

In Section~\ref{sec:method:construct}, we have introduced structural schemas for $\mathsf{DupCopy}$ and $\mathsf{DupClone}$. These are the sole sources of unbounded token growth. Under bound $B$, a duplication transition $t_{dup}$ for type $\tau$ is enabled at marking $M$ only if: $|M(p^{\mathsf{own}}_\tau)| < B(p^{\mathsf{own}}_\tau)$. This constraint is naturally enforced by the global definition of bounded reachability: $\Cfg \xrightarrow{\varphi}_{B,D} \Cfg'$ iff $\Cfg \xrightarrow{\varphi} \Cfg'$ and $\mathsf{BudgetOK}_{B,D}(\Cfg')$, where $\mathsf{BudgetOK}_{B,D}(\langle M, S \rangle) \equiv (\forall p. |M(p)| \le B(p)) \wedge |S| \le D$. Write $\Cfg \xrightarrow{\varphi}_{B,D} \Cfg'$ if $\Cfg[\varphi\rangle \Cfg'$ (Def.~\ref{def:pcpn:firing})
and $\mathsf{BudgetOK}_{B,D}(\Cfg')$.
Let $\Rightarrow^{*}_{B,D}$ be the reflexive transitive closure of $\xrightarrow{\ }_{B,D}$.

\subsection{Canonicalization and finite branching}
\label{sec:analysis:canon}
Reachability is invariant under bijective renaming of value identifiers and region labels (Lemma~\ref{lem:beta-inv}). Consequently, we perform the analysis over canonical representatives of $\beta$-equivalence classes, computed via a canonicalization function.

\begin{definition}[Canonicalization]
\label{def:canon}
A canonicalization function $\Canon:\Cfg\to\Cfg$ maps a configuration to a representative of its $\beta$-equivalence class by: i) renaming value identifiers in the order of their occurrence to a fixed canonical pool; ii) renaming region labels in the order of their occurrence when scanning the stack from top to bottom; and iii) sorting token multisets within each place by a fixed total order on colors.
\end{definition}

Because $\widehat{\mathsf{Ty}}$ and the sets of signature variables $\mathsf{TVar},\mathsf{LVar}$ are finite, and because the assumed boundedness of the system limits the number of simultaneously live identifiers and labels, the search space of unification completions is finite. In the implementation, we additionally choose the least fresh identifier or label with respect to fixed total orders, making freshness choices deterministic.

\subsection{Canonical Bounded Reachability Graph}
\begin{theorem}[Finiteness of Reachable States]
\label{thm:canon-reach}
The set of canonical bounded reachable states, defined as $\widehat{V}_{B,D}
~:=~
\left\{\ \Canon(\Cfg)\ \middle|\ \langle M_0,\epsilon\rangle \Rightarrow^{*}_{B,D} \Cfg\ \right\}$,
is finite.
\end{theorem}

\begin{proof}
By the bounding condition $\mathsf{BudgetOK}_{B,D}$, every reachable configuration contains at most $\sum_{p\in P}B(p)$ tokens and at most $D$ frames. The set of places $P$ and the set of ground types $\widehat{\mathsf{Ty}}$ are finite. The canonicalization function abstracts away concrete identifier and label names, leaving only finitely many canonical shapes under the given bounds. \hfill $\blacksquare$
\end{proof}

\begin{algorithm}[t]
\caption{Worklist saturation of $\widehat{G}_{B,D}$ with a parent map}
\label{alg:buildmrg}
\DontPrintSemicolon
\KwIn{PCPN $\mathcal{N}_{\mathcal{C}}$, bounds $(B,D)$}
\KwOut{Graph $\widehat{G}_{B,D}$ and parent map $\mathsf{par}$}

$\widehat{\Cfg}_0 \leftarrow \Canon(\langle M_0,\epsilon\rangle)$\;
$Q \leftarrow [\widehat{\Cfg}_0]$; $\widehat{V} \leftarrow \{\widehat{\Cfg}_0\}$; $\widehat{E} \leftarrow \emptyset$; $\mathsf{par}(\widehat{\Cfg}_0) \leftarrow \bot$\;

\While{$Q \neq \emptyset$}{
    $\widehat{\Cfg} \leftarrow \mathrm{pop}(Q)$\;
    \ForEach{$t \in T$}{
        \ForEach{valid inputs $\chi \subseteq \widehat{\Cfg}.M$}{
            \ForEach{$\iota \in \SolveGuard(t,\widehat{\Cfg},\chi)$}{
                \If{$t$ is enabled at $\widehat{\Cfg}$ under $(\chi,\iota)$}{
                    $\Cfg' \leftarrow \Fire(\widehat{\Cfg},\langle t,\chi,\iota\rangle)$\;
                    \If{$\mathsf{BudgetOK}_{B,D}(\Cfg')$}{
                        $\widehat{\Cfg}' \leftarrow \Canon(\Cfg')$\;
                        $\widehat{E} \leftarrow \widehat{E} \cup \{(\widehat{\Cfg},\langle t,\chi,\iota\rangle,\widehat{\Cfg}')\}$\;
                        \If{$\widehat{\Cfg}' \notin \widehat{V}$}{
                            $\widehat{V} \leftarrow \widehat{V} \cup \{\widehat{\Cfg}'\}$;$\mathrm{push}(Q,\widehat{\Cfg}')$\;
                            $\mathsf{par}(\widehat{\Cfg}') \leftarrow (\widehat{\Cfg},\langle t,\chi,\iota\rangle)$\;
                        }
                    }
                }
            }
        }
    }
}
\Return{$(\widehat{G}_{B,D}=(\widehat{V},\widehat{E}),\ \mathsf{par})$}\;
\end{algorithm}

\begin{definition}[Canonical bounded reachability graph]
\label{def:cbg}
The canonical bounded reachability graph is the directed labeled graph
$\widehat{G}_{B,D}=(\widehat{V}_{B,D},\widehat{E}_{B,D})$ where
$(\widehat{\Cfg},\varphi,\widehat{\Cfg}')\in \widehat{E}_{B,D}$ if
there exist configurations $\Cfg,\Cfg'$ such that
$\Canon(\Cfg)=\widehat{\Cfg}$, $\Canon(\Cfg')=\widehat{\Cfg}'$, and $\Cfg \xrightarrow{\varphi}_{B,D} \Cfg'$.
\end{definition}

Algorithm~\ref{alg:buildmrg} constructs $\widehat{G}_{B,D}$ via worklist saturation. Starting from $\widehat{\Cfg}_0$, it iteratively expands states, and valid successors (checked via $\mathsf{BudgetOK}$) are canonicalized and added to the graph if new, mapping parents in $\mathsf{par}$ for witness extraction.

\begin{theorem}[Termination and completeness within bounds]
\label{thm:mrg-term-comp}
Algorithm~\ref{alg:buildmrg} terminates and returns the complete graph $\widehat{G}_{B,D}$.
Moreover, for every node $\widehat{\Cfg}\in\widehat{V}_{B,D}$ with $\mathsf{par}(\widehat{\Cfg})\neq\bot$,
backtracking $\mathsf{par}$ yields a witness firing sequence from $\widehat{\Cfg}_0$ to $\widehat{\Cfg}$.
\end{theorem}
\begin{proof}
Termination follows from finiteness of $\widehat{V}_{B,D}$ (Theorem~\ref{thm:canon-reach}) and the fact that each new node is enqueued once.
Completeness follows because the algorithm enumerates all bounded-enabled firings from every discovered node and adds the corresponding edges \cite{jensen1987coloured,liu2022petri}.
The parent map is set on the first discovery of each node, and ths backtracking yields a valid path. \hfill $\blacksquare$
\end{proof}
Once the reachability graph is constructed, we synthesize the final executable Rust code using the procedure outlined in Algorithm~\ref{alg:synthesize}.
This algorithm transforms a graph path into a valid sequence of operations that satisfies the synthesis objective.

\begin{algorithm}[t]
\caption{Goal-directed synthesis over $\widehat{G}_{B,D}$}
\label{alg:synthesize}
\DontPrintSemicolon
\KwIn{MRG $(\widehat{G}_{B,D},\mathsf{par})$, Objective $\mathsf{Goal}$}
\KwOut{Rust Snippet satisfying $\mathsf{Goal}$, or $\bot$}

Select $\widehat{\Cfg}_g \in \widehat{V}_{B,D}$ such that $\mathsf{Goal}(\widehat{\Cfg}_g)$\;
\lIf{no such $\widehat{\Cfg}_g$ exists}{\Return $\bot$}

$\sigma \leftarrow \mathsf{Backtrack}(\mathsf{par},\widehat{\Cfg}_g)$\;
Let $\Cfg_g$ be the concrete state reached by $\sigma$ from $\langle M_0, \epsilon \rangle$\;
$\sigma_{\mathrm{close}} \leftarrow \textsc{CloseStack}(\Cfg_g)$; \lIf{$\sigma_{\mathrm{close}}=\bot$}{\Return $\bot$}

$\sigma^\downarrow \leftarrow \sigma \cdot \sigma_{\mathrm{close}}$\;
\Return{$\mathsf{Emit}(\sigma^\downarrow)$}\;
\end{algorithm}
Algorithm~\ref{alg:synthesize} first scans the vertices $\widehat{V}_{B,D}$ to find a candidate state $\widehat{\Cfg}_g$ that meets the user-defined $\mathsf{Goal}$. It then recovers the execution trace $\sigma$ by backtracking via the parent map.
However, this trace may end inside nested scopes. To ensure syntactic validity (e.g., closing open scopes), it invokes Algorithm~\ref{alg:closestack}, which iteratively fires unique scope-ending transitions until the stack empties ($S=\epsilon$), yielding a complete sequence $\sigma^\downarrow$.

\subsection{Witness extraction, deterministic stack closure, and emission}
\label{sec:analysis:emit}

\begin{definition}[Backtracking a witness]
\label{def:backtrack}
Given a parent map $\mathsf{par}$ and a node $\widehat{\Cfg}$, define $\mathsf{Backtrack}(\mathsf{par},\widehat{\Cfg})$ as:
$\mathsf{Backtrack}(\mathsf{par},\widehat{\Cfg}_0)=\epsilon$, and if
$\mathsf{par}(\widehat{\Cfg})=(\widehat{\Cfg}_p,\varphi)$ then
$\mathsf{Backtrack}(\mathsf{par},\widehat{\Cfg})=\mathsf{Backtrack}(\mathsf{par},\widehat{\Cfg}_p)\cdot \varphi$.
\end{definition}

A configuration is \emph{closed} if its stack is empty: $\mathsf{Closed}(\langle M,S\rangle)\triangleq (S=\epsilon)$.
Given a witness reaching $\Cfg$, we optionally append a closure suffix that repeatedly pops the top frame.

\begin{algorithm}[t]
\caption{$\textsc{CloseStack}(\Cfg)$}
\label{alg:closestack}
\DontPrintSemicolon
\KwIn{Configuration $\Cfg=\langle M,S\rangle$}
\KwOut{Closing sequence $\sigma_{\mathrm{close}}$, or $\bot$}

$\sigma_{\mathrm{close}} \leftarrow \epsilon$\;
\While{$S \neq \epsilon$}{
    Let $\gamma = \mathrm{top}(S)$\;
    Find unique transition $t$ enabling $\mathsf{pop}(\gamma)$
    \lIf{no such transition exists}{\Return $\bot$}
    
    $\Cfg \leftarrow \Fire(\Cfg, t)$;$\sigma_{\mathrm{close}} \leftarrow \sigma_{\mathrm{close}} \cdot t$;$S \leftarrow \Cfg.S$\;
}
\Return{$\sigma_{\mathrm{close}}$}\;
\end{algorithm}

\begin{theorem}[Soundness of bounded synthesis]
\label{thm:synth-sound}
If Algorithm~\ref{alg:synthesize} returns $s\neq\bot$, then there exists a closed bounded firing sequence $\sigma^\downarrow$
such that $s=\mathsf{Emit}(\sigma^\downarrow)$ and
$\langle M_0,\epsilon\rangle[\sigma^\downarrow\rangle_{B,D}\Cfg$ for some closed configuration $\Cfg$.
Consequently, the emitted snippet is accepted by the SRS equivalently, by the guarded signature checks.
\end{theorem}
\begin{proof}
By construction, $\sigma$ is a bounded witness to $\widehat{\Cfg}_g$ in the canonical graph.
If $CloseStack$ succeeds, $\sigma^\downarrow$ reaches a closed configuration by repeated stack pops.
Each firing in $\sigma^\downarrow$ is enabled by guarded solving, hence corresponds to a valid SRS step, and hence $\mathsf{Emit}(\sigma^\downarrow)$ is a valid SRS trace realization. \hfill $\blacksquare$
\end{proof}


\section{Conclusion}
\label{sec:conlusion}
This paper presented a semantics-first approach to signature-level Rust code generation, establishing a strong bisimulation between SRS and PCPN to characterize compilable traces. Under explicit bounds, we construct a finite reachability graph for deterministic witness extraction. However, our analysis is an under-approximation limited by the specified budget and potential state explosion in complex generic dependencies. Future work includes extending the theory to \textit{unsafe} and \textit{async} features, optimizing graph exploration via guided heuristics, and applying the generator to differential compiler testing.

\bibliographystyle{splncs04}
\bibliography{pcpn}

\end{document}